\documentclass[letterpaper, 10 pt, conference]{IEEEtran}

\usepackage{cite}

\usepackage[cmex10]{amsmath}
\usepackage{amssymb,amsfonts}
\usepackage[hyphens]{url}
\usepackage{algorithmic,algorithm}
\usepackage{amsthm}

\usepackage{graphicx}
\usepackage[caption=false,font=footnotesize]{subfig}
\usepackage{multirow}
\usepackage{etoolbox}

\newtoggle{isreport}
\toggletrue{isreport}

\newtheorem{lemma}{Lemma}[section]

\newtheorem{cor}[lemma]{Corollary}
\newtheorem{thm}[lemma]{Theorem}

\newcommand{\qedd}{\hfill{$\blacksquare$}}

\newcommand{\paren}[1]{\left(#1\right)}
\newcommand{\bracket}[1]{\left[#1\right]}
\newcommand{\set}[1]{\left\{#1\right\}}
\newcommand{\abs}[1]{\left|#1\right|}
\newcommand{\norm}[1]{\left\|#1\right\|}

\newcommand{\bff}[1]{{\bf #1}}

\DeclareMathOperator{\rank}{rank}
\DeclareMathOperator{\kernel}{kernel}

\newcommand{\R}{\mathbb R}
\newcommand{\C}{\mathbb C}

\newcommand{\ol}[1]{\overline{#1}}

\newcommand{\calE}{\mathcal{E}}

\newcommand{\calG}{\mathcal{G}}
\newcommand{\calH}{\mathcal{H}}

\newcommand{\calL}{\mathcal{L}}

\newcommand{\calN}{\mathcal{N}}

\usepackage{color}
\usepackage{comment}
\usepackage{ifthen}
\newboolean{showcomments}
\setboolean{showcomments}{true}
\newcommand{\lguo}[1]{\ifthenelse{\boolean{showcomments}}
{ \textcolor{red}{(Daniel says:  #1)}}{}}
\newcommand{\slow}[1]{\ifthenelse{\boolean{showcomments}}
{ \textcolor{blue}{(Steven says:  #1)}}{}}

\begin{document}
\title{Graph Laplacian Spectrum and Primary Frequency Regulation}

\IEEEoverridecommandlockouts

\IEEEoverridecommandlockouts
\author{Linqi Guo, Changhong Zhao, and 
	Steven H.~Low \thanks{Linqi Guo and Steven H.~Low are with the Department of Computing and Mathematical Sciences, California Institute of Technology, Pasadena, CA, 91125. Email: \texttt{\{lguo, slow\}@caltech.edu}. Changhong Zhao is with the National Renewable Energy Laboratory, Golden, CO, 80401. Email: \texttt{Changhong.Zhao@nrel.gov}.}
\thanks{The authors thank Janusz Bialek and Oleg Khamisov from Skoltech for helpful discussions. This work has been supported by Resnick Research Fellowship, Linde Institute Research Award, DOE through the ENERGISE program (Award \#DE-EE-0007998), NSF grants through CCF 1637598, ECCS 1619352, CNS 1545096, ARPA-E grant through award DE-AR0000699 (NODES) and GRID DATA, DTRA through grant HDTRA 1-15-1-0003 and Skoltech through collaboration agreement 1075-MRA.}}
\maketitle

\begin{abstract}
We present a framework based on spectral graph theory that captures the interplay among network topology, system inertia, and generator and load damping in determining the overall grid behavior and performance. Specifically, we show that the impact of network topology on a power system can be quantified through the network Laplacian eigenvalues, and such eigenvalues determine the grid robustness against low frequency disturbances. Moreover, we can explicitly decompose the
frequency signal along scaled Laplacian eigenvectors when damping-inertia ratios are uniform across buses. The insight revealed by this framework partially explains why load-side participation in frequency regulation not only makes the system respond faster, but also helps lower the system nadir after a disturbance. Finally, by presenting a new controller specifically tailored to suppress high frequency disturbances, we demonstrate that our results can provide useful guidelines in the controller design for load-side primary frequency regulation. This improved controller is simulated on the IEEE 39-bus New England interconnection system to illustrate its robustness against high frequency oscillations compared to both the conventional droop control and a recent controller design.

\end{abstract}

\IEEEpeerreviewmaketitle

\section{Introduction}\label{section:intro}
Frequency regulation balances the power generation and consumption in an electric grid. Such control is governed by the swing dynamics and is traditionally implemented in generators through droop control, automatic generation control and economic dispatch \cite{zhao2014design, zhao2016unified}. It has been widely realized in the community that the increasing level of renewable penetration makes it harder to stabilize the system due to higher generation volatility and lower aggregate inertia. One popular approach to maintaining system stability in this new era is to integrate load-side participation \cite{kirby2003spinning,pnnl2012grid,trudnowski2006decentralized,short2007stabilization,
callaway2011achieving,molina2011decentralized,andreasson2014distributed,dorfler2017gather,bouattour2013further}, which not only helps stabilize the system in a more responsive and scalable fashion, but also improves the system transient behavior \cite{zhao2012swing,zhao2014design,zhao2016unified}.  We redirect the readers to \cite{molzahn2017survey} for an extensive survey on recent frequency regulation controller designs.

The benefits of load-side controllers motivated a series of work on understanding how different system parameters and controller designs impact the grid transient performance. For instance, iDroop is proposed in \cite{mallada2016idroop} to improve dynamic performance of the power system through controlling power electronics or loads. Such controllers, however, can sometimes make the power system dynamics more sophisticated and uncertain and hence make it harder to obtain a stability guarantee \cite{pates2016decentralized}. In \cite{poola2017optimal}, methods to determine the optimal placement of virtual inertia in power grids to accommodate loss of system stability are proposed and studied. There has also been work on characterizing the synchronization cost of the swing dynamics \cite{dorfler2010spectral,tegling2015price,pirani2017system,paganini2017global,coletta2017performance} that explicitly computes the response $\calH_2$ norm in terms of system damping, inertia, resistive loss, line failures etc. In certain cases, classical metrics studied in power engineering such as nadir and maximum rate of change of frequency can also be analytically derived \cite{paganini2017global}.

Compared to the aforementioned system parameters, the role of transmission topology on the transient stability of swing dynamics is less well understood. Indeed, it is usually hard to infer without detailed simulation and computation on how a change to the network topology affects overall grid behaviour and performance. For example, one can argue that the connectivity in the grid helps average the power demand imbalance over the network and therefore higher connectivity should enhance system stability. On the other hand, one can also argue that higher connectivity means faster propagation of disturbances over the network, and should therefore decrease system stability. Both arguments seem plausible but they lead to (apparently) opposite conclusions (a corollary of our results in Section \ref{section:characterization} will clarify this paradox). In fact, even the notion ``connectivity'' itself seems vague and can be interpreted in different ways.

In this work, we present a framework based on spectral graph theory that captures the interplay among network topology, system inertia, and generator and load damping. Compared to existing literature \cite{mallada2016idroop,pates2016decentralized,poola2017optimal,dorfler2010spectral,tegling2015price,pirani2017system,paganini2017global,coletta2017performance} that usually relies on system norms in the analysis, our approach studies more directly the frequency trajectory itself through a decomposition along scaled Laplacian spectrum. This allows us to arrive at certain insights that are either difficult to derive or are overlooked within existing methods. Our contributions can be summarized as follows: a) We show that whether the system oscillates or not is determined by how strong the damping normalized by inertia is compared to network connectivity in the ``corresponding'' direction; b) We prove that the power grid robustness against low frequency disturbance is mostly determined by network connectivity, while its robustness against high frequency disturbance is mostly determined by system inertia; c) We demonstrate that although increasing system damping helps suppressing disturbences, such benefits are mostly in the medium frequency band; d) We devise a quantitative explanation on why load-side participation helps improve system behavior in the transient state, and demonstrate how our results suggest an improved controller design that can suppress input noise much more effectively.

The rest of this paper is organized as follows. In Section \ref{section:model}, we review the system model and relevant concepts from spectral graph theory. In particular, we provide a rigorous definition on the ``strength'' of connectivity. In Section \ref{section:characterization}, we present our characterization of the system response in both the time and Laplace domain. The practical interpretations of our results are given in Section \ref{section:interpretations}. In Section \ref{section:design}, we quantify the benefits of load-side controllers and present a new controller that is specifically tailored to suppress high frequency oscillation. In Section \ref{section:evaluation}, we simulate the improved controller on the IEEE 39-bus New England interconnection testbed and illustrate its robustness against measurement noise and high frequency oscillation in injection. We conclude in Section \ref{section:conclusion}.

\section{Network Model}\label{section:model}
In this section, we present the system model as adopted in \cite{zhao2014design} and review relevant concepts from spectral graph theory.

Let $\R$ and $\C$ denote the set of real and complex numbers, respectively. For two matrices $A, B$ with proper dimensions, $[A~B]$ means the concatenation of $A, B$ in a row, and $[A; B]$ means the concatenation of $A, B$ in a column. A variable without subscript usually denotes a vector with appropriate components, e.g., $\omega=(\omega_j, j\in\calN)\in\R^{\abs{\calN}}$. For a time-dependent signal $\omega(t)$, we use $\dot{\omega}$ to denote its time derivative $\frac{d\omega}{dt}$. The identity matrix of dimension $n\times n$ is denoted as $I_n$. The column vector of length $n$ with all entries being $1$ is denoted as $\bff{1}_n$. The imaginary unit $\sqrt{-1}$ is denoted as $\mathrm{j}$.

We use a weighted graph $\calG=(\calN, \calE)$ to describe the power transmission network, where $\calN=\set{1,\ldots, n}$ is the set of buses and $\calE\subset\calN \times \calN$ denotes the set of transmission lines weighted by its line susceptances. The terms bus/node and line/edge are used interchangeably in this paper. We assume without loss of generality that $\calG$ is connected and simple. An edge in $\calE$ is denoted either as $e$ or $(i,j)$. We further assign an arbitrary orientation over $\calE$ so that if $(i,j)\in\calE$ then $(j,i)\notin\calE$.

Let $n, m$ be the number of buses and transmission lines respectively. The incidence matrix of $\calG$ is the $n\times m$ matrix $C$ defined as
$$
C_{je}=\begin{cases}
  1 & \text{if node }j\text{ is the source of }e\\
  -1 & \text{if node }j\text{ is the target of }e\\
  0 &\text{otherwise}
\end{cases}
$$
For each bus $j\in\calN$, we denote its frequency deviation
as $\omega_j$ and denote the inertia constant as $M_j>0$. The symbol $P_j^m$ is overloaded to denote the mechanical power injection if $j$ is a generator bus and denote the aggregate power injection from uncontrollable loads if $j$ is a load bus. For a generator bus, we absorb the droop control into a damping term $D_j\omega_j$ with $D_j\ge 0$ and for load buses, we use the same symbol to denote the aggregated frequency sensitive load. For each transmission line $(i,j)\in\calE$, denote as $P_{ij}$ the branch flow deviation and denote as $B_{ij}$ the line susceptance assuming voltage magnitudes are 1 p.u. With such notations, the linearized swing and network dynamics are given by
\begin{subequations}\label{eqn:swing_and_network_dynamics}
\begin{IEEEeqnarray}{rCll}
  M_j\dot{\omega}_j&=&-D_j\omega_j - d_j + P_j^m - \sum_{e\in\calE}C_{je}P_e,&\quad j\in\calN\\
  \dot{P}_{ij}&=&B_{ij}(\omega_i-\omega_j),\quad& (i,j)\in\calE \label{eqn:network_flow_dynamics}
\end{IEEEeqnarray}
\end{subequations}
We refer the readers to \cite{zhao2014design} for more detailed justification and derivation of this model. 

Now using $x$ to denote the system state $x=[\omega; P]$, and putting $M$, $D$ and $B$ to be the diagonal matrices with $M_j$, $D_j$ and $B_{ij}$ as diagonal entries respectively, we can rewrite the system dynamics \eqref{eqn:swing_and_network_dynamics} in the state-space form
\begin{IEEEeqnarray}{rCl}
\dot{x}&=&\begin{bmatrix}
  -M^{-1}D & -M^{-1}C \\
  BC^T & 0
\end{bmatrix}x +
\begin{bmatrix}
  M^{-1} \\ 0
\end{bmatrix}(P^m-d)\label{eqn:state_space}
\end{IEEEeqnarray}
The matrix
$$
A=\begin{bmatrix}
  -M^{-1}D & -M^{-1}C \\
  BC^T & 0
\end{bmatrix}
$$
is referred to as the system matrix in the sequel. The system \eqref{eqn:state_space} can be interpreted as a multi-input-multi-output linear system with input $P^m-d$ and output $x$. We emphasize that the variables $[\omega;P]$ denote deviations from their nominal values so that $x(t)=0$ means the system is in its nominal state at time $t$.

For any node $i\in\calN$, we denote the set of its neighbors as $N(i)$. The (scaled) graph Laplacian matrix of $\calG$ is the $n\times n$ symmetric matrix $L=M^{-1/2}CBC^TM^{-1/2}$, which is explicitly given by
$$
L_{ij}=\begin{cases}
-\frac{B_{ij}}{\sqrt{M_iM_j}} & i\neq j, (i,j)\in\calE \text{ or } (j,i)\in\calE \\
\frac{1}{M_i}\sum_{j:j\in N(i)}B_{ij} & i=j, \abs{N(i)}>0\\
0 & \text{otherwise}
\end{cases}
$$
It is well known that if the graph $\calG$ is connected, then $L$ has rank $n-1$, and any principal minor of $L$ is invertible \cite{chung1997spectral}. For any vector $x\in\R^{n}$, we have
$$
x^TLx=\sum_{(i,j)\in\calE}B_{ij}\paren{\frac{x_i}{\sqrt{M_i}}-\frac{x_j}{\sqrt{M_j}}}^2\ge 0
$$
This implies that $L$ is a positive semidefinite matrix and thus diagonalizable. We denote its eigenvalues and corresponding orthonormal eigenvectors as
$0=\lambda_1 < \lambda_2\le\cdots\le \lambda_n$ and $v_1, v_2,\cdots,v_n$. When the matrix $L$ has repeated eigenvalues, for each repeated eigenvalue $\lambda_i$ with multiplicity $m_i$, the corresponding eigenspace of $L$ always has dimension $m_i$, hence an orthonormal basis consisting of eigenvectors of $L$ exists (yet such bases are not unique). We assume one of the possible orthonormal bases is chosen and fixed throughout the paper.

The eigenvalues of the graph Laplacian matrix measure the graph connectivity from an algebraic perspective, and larger Laplacian eigenvalues suggest stronger connectivity.  To make such discussion more concrete, we define a partial order $\preceq$ over the set of all weighted graphs (possibly disconnected) with vertex set $\calN$ as follows: For two weighted graphs $\calG_1=(\calN, \calE_1)$ and $\calG_2=(\calN, \calE_2)$, we say $\calG_1\preceq \calG_2$ if $\calE_1\subset \calE_2$, and for any $e\in\calE_1$, the weight of $e$ in $\calG_1$ is no larger than that in $\calG_2$. It is routine to check that $\preceq$ defines a partial order\footnote{We emphasize that this is not a complete order over all graphs. That is, not any pair of graphs of the same number of vertices are comparable through this order.}. A more interesting result is that the mapping from a graph to its Laplacian eigenvalues preserves this order:
\begin{lemma}
Let $L_1$ and $L_2$ be the (scaled) Laplacian matrices of two weighted graphs $\calG_1$ and $\calG_2$ with $\calG_1\preceq \calG_2$. Let $0=\lambda_1^1\le \lambda_2^1\le \cdots \le \lambda_n^1$ and $0=\lambda_1^2\le \lambda_2^2\le \cdots \le \lambda_n^2$ be the eigenvalues of $L_1$ and $L_2$ respectively. Then
$$
\lambda_i^1 \le \lambda_i^2, i = 1, 2, \ldots, n
$$
\end{lemma}
\begin{proof}
This result follows from the fact that $L_2-L_1$ is positive semidefinite, which is easy to check from our definition of $\preceq$.
\end{proof}
In fact we can devise better estimates on the relative orders of the eigenvalues $\lambda_i^j$. Interested readers are refered to \cite{guo2017monotonicity} for more discussions and results therein. Throughout this paper, whenever we compare two graphs in terms of their connectivity, we always refer to the partial order $\preceq$.

In the sequel, we further assume that the inertia and damping of the buses are proportional to its power ratings. That is, we assume there is a baseline inertia $\mu$ and damping $\delta$ such that for each generator $j$ with power rating $f_j$, we have $M_j=f_j \mu$ and $D_j=f_j \delta$. This is a natural setting as machines with high ratings are typically ``heavy'' and have more significant impact on the overall system dynamics. See \cite{paganini2017global, coletta2017performance} for more details. Under such assumptions, the ratios $D_j/M_j$ is independent of $j$, and therefore $M^{-1}D=\gamma I_n$ where $\gamma=\delta/\mu>0$. We will study both the transmission graph Laplacian matrix and Laplace domain properties of \eqref{eqn:state_space}. To clear potential confusion, we agree that whenever the adjective Laplacian is used, we refer to quantities related to the Laplacian matrix $L$, while whenever the noun Laplace is used, we refer to notions about the Laplace transform
$$
\calL\set{s(t)}(\tau):=\int_0^\infty s(t)e^{-\tau t} dt
$$
or notions defined in the Laplace domain.

\section{Characterization of System Response}\label{section:characterization}
In this section, we give a complete characterization of the system response of \eqref{eqn:state_space} based on spectral decomposition in both time and Laplace domain.

\subsection{Stability under zero input}
We first determine the modes of the system \eqref{eqn:state_space}. That is, we compute the eigenvalues of the system matrix $A$. Such eigenvalues indicate whether the system is stable and if it is, how fast the system converges to an equilibrium state.

\begin{thm}\label{thm:modes}
Let $0=\lambda_1< \lambda_2\le \cdots \le \lambda_n$ be the eigenvalues of $L$ with corresponding orthonormal eigenvectors $v_1, v_2, \ldots, v_n$. Then:
\begin{enumerate}
\item $0$ is an eigenvalue of $A$ of multiplicity $m-n+1$. The corresponding eigenvectors are of the form $[0; P]$ with $P\in\kernel(C)$
\item $-\gamma$ is a simple eigenvalue of $A$ with $\bracket{M^{-1/2}v_1;0}$ as a corresponding eigenvector
\item For $i=2,3,\ldots, n$, $\phi_{i,\pm}=\frac{-\gamma \pm\sqrt{\gamma^2-4\lambda_i}}{2}$ are eigenvalues of $A$. For any such $\phi_{i,\pm}$, an eigenvector is given by $\bracket{M^{-1/2}v_i; \phi_{i,\pm}^{-1}BC^TM^{-1/2}v_i}$.
\end{enumerate}
\end{thm}

\begin{proof}
\iftoggle{isreport}{
Please refer to Appendix \ref{proof:modes}.
}{
Please refer to our online report \cite{report}.
}
\end{proof}

When $m-n+1=0$ or equivalently when the network is a tree, item 1) of Theorem \ref{thm:modes} is understood to mean that the system matrix $A$ does not have $0$ as an eigenvalue. We remark that a similar characterization of the system \eqref{eqn:state_space} under different state representation can be found in  \cite{coletta2017performance}.

 Assuming $\gamma^2-4\lambda_i\neq 0$ for all $i$, we get $2n-1$ nonzero eigenvalues of $A$ from item 2) and item 3) of Theorem \ref{thm:modes}, counting multiplicity, which together with the $m-n+1$ multiplicity from item 1) gives $m+n$ eigenvalues as well as $m+n$ linearly independent eigenvectors. Therefore we know $A$ is always diagonalizable over the complex field $\C$, provided critical damping, that is $\gamma^2-4\lambda_i=0$ for some $i$, does not occur. We assume this is the case in all following derivations. When critical damping does occur, our results can be generalized using the standard Jordan decomposition.

Theorem \ref{thm:modes} explicitly reveals the impact of the transmission network connectivity as captured by its Laplacian eigenvalues on the system \eqref{eqn:state_space} and tells us that the system mode shape is closely related to the corresponding Laplacian eigenvectors. In particular, we note that the real parts of $\phi_{i,\pm}$ are nonpositive, from which we deduce the following corollary.

\begin{cor}\label{cor:stability}
The system \eqref{eqn:state_space} is marginally stable, with marginal stable states of the form $[0;P]$ with $P\in\kernel(C)$. Therefore the system \eqref{eqn:state_space} is asymptotically stable on a tree.
\end{cor}

The kernel of $C$ corresponds to the set of branch flow vectors $P$ such that $\sum_{j\in N(i)} P_{ij}=0$ for all $i\in\calN$. They can be interpreted as flows that are balanced at all the buses (e.g.~ circulation flows on a loop) for which each bus $i$ is neither a source node ($\sum_{j\in N(i)} P_{ij} >0$) nor a sink node ($\sum_{j\in N(i)} P_{ij} <0$). This corollary tells us that the only possible signals that can persist in \eqref{eqn:state_space} are the balancing branch flows. Of course, such marginally stable flows cannot exist in a real system because of losses in transmission lines (in which case our network dynamics \eqref{eqn:network_flow_dynamics} is no longer accurate). Even if we take the simplified model \eqref{eqn:state_space}, as long as the initial system branch flow does not belong to $\kernel(C)$, the system \eqref{eqn:state_space} under zero input $P^m-d=0$ converges to the nominal state.
%

\subsection{System response to step input}
Next we determine the system response to a step function. More precisely, we define $s(t):=P^m(t)-d(t)$ as the surplus function and compute the frequency trajectory $\omega(t)$ with $s(t)$ as input to \eqref{eqn:state_space}, assuming $s(t)$ takes constant value $s$ over time. The components $s_j$ can be different over $j$. We put $s=\sum_{i}\hat{s}_iM^{1/2}v_i$ to be the decomposition of $s$ along the scaled Lapalacian eigenvectors (note the decomposition scaling $M^{1/2}v_i$ is different from the scaling $M^{-1/2}v_i$ in the following theorem statement).

\begin{thm}\label{thm:response}
Let $0=\lambda_1< \lambda_2\le \cdots\le \lambda_n$ be the eigenvalues of $L$ with corresponding orthonormal eigenvectors $v_1, v_2,\ldots, v_n$. Assume:
\begin{enumerate}
\item The system \eqref{eqn:state_space} is initially at the nominal state $x(0)=0$
\item $\gamma^2-4\lambda_i\neq 0$ for all $i$.
\end{enumerate}
Then
\begin{IEEEeqnarray}{rCl}
\omega(t)
&=&\sum_{i=1}^n\frac{\hat{s}_i}{\sqrt{\gamma^2-4\lambda_i}}\paren{e^{\phi_{i,+}t}-e^{\phi_{i,-}t}}M^{-1/2}v_i \label{eqn:omega_evolution}
\end{IEEEeqnarray}
where
$$
\phi_{i,+}:=\frac{-\gamma+\sqrt{\gamma^2-4\lambda_i}}{2} \quad \phi_{i,-}:=\frac{-\gamma-\sqrt{\gamma^2-4\lambda_i}}{2}
$$
$ $  
\end{thm}

\begin{proof}
\iftoggle{isreport}{
Please refer to Appendix \ref{proof:response}.
}{
Please refer to our online report \cite{report}.
}
\end{proof}

We remark that all conditions in this theorem are for presentation simplicity and the frequency trajectory \eqref{eqn:omega_evolution} can be generalized by adding correction terms to the case where neither condition is imposed. We opt to not doing so here as these terms lead to more tedious notations yet do not reveal any new insights. 

This result tells us that the frequency trajectory of \eqref{eqn:state_space} can be decomposed along scaled eigenvectors of the Laplacian matrix $L$. Moreover, we note that all $\phi_{i,\pm}$ have negative real parts except $\phi_{1,+}=0$. Therefore the only term in \eqref{eqn:omega_evolution} that persists is the term involving $\phi_{1,+}$ given as:
\begin{IEEEeqnarray*}{rCl}
&&\frac{\hat{s}_1}{\sqrt{\gamma^2-4\lambda_1}}e^{\phi_{1,+}t}M^{-1/2}v_1=\frac{\hat{s}_1}{\gamma}M^{-1/2}v_1
\end{IEEEeqnarray*}
Thus under the input $s=P^m-d$, the $\omega(t)$ signal converges to the steady state $\frac{\hat{s}_1}{\gamma}M^{-1/2}v_1$ exponentially fast. This allows us to recover the following well-known result in frequency-regulation literature \cite{janusz2008book}, using a new argument.

\begin{cor}\label{cor:balancing_injection}
Under step input $s$, the system \eqref{eqn:state_space} converge to a steady state with synchronized frequencies $\omega_i=\omega_j=:\omega_c$. Moreover, $\omega_c=0$ if and only if the power injection is balanced $\sum_{i\in\calN}s_i = 0$.
\end{cor}
\begin{proof}
It is easy to show 
$$v_1=\frac{M^{1/2}}{\sqrt{\sum_{j\in\calN} M_j}}\bff{1}_n$$
By Theorem \ref{thm:response}, we know the steady state of \eqref{eqn:state_space} is $(\hat{s}_1/\gamma)M^{-1/2}v_1$, which then has all entries equal to the same value 
$$\frac{\hat{s}_1}{\gamma\sqrt{\sum_j M_j}}$$
Therefore $\omega_i=\omega_j=:\omega_c$ for all $i,j\in\calN$. From $s=\sum_{i}\hat{s}_iM^{1/2}v_i$ we see $\hat{s}_1=(M^{-1/2}s)^Tv_1=s^TM^{-1/2}v_1$, and thus
\begin{IEEEeqnarray*}{rCl}
\sum_{i\in\calN}s_i &=& s^T\bff{1}_n=\sqrt{\sum_jM_j}s^TM^{-1/2}v_1=\sqrt{\sum_jM_j}\hat{s}_1\\
&=&\gamma\paren{\sum_jM_j}\omega_c=\paren{\sum_{j\in\calN}D_j}\omega_c
\end{IEEEeqnarray*}
Hence $\omega_c=0$ if and only if $\sum_{i\in\calN}s_i=0$.
\end{proof}


\subsection{Spectral transfer functions for arbitrary input}
It is also informative to look at the system behavior of \eqref{eqn:state_space} from the Laplace domain. Instead of analyzing transfer functions from any input to any output as in the classical multi-input-multi-output system analysis, we take a slightly different approach such that the Laplacian matrix spectral information is preserved. More precisely, for a time-variant surplus signal $s(t)$, we first decompose it to the spectral representation $s(t)=\sum_{i=1}^n\hat{s}_i(t)M^{1/2}v_i$. Now $\hat{s}_i(t)$ is a real-valued signal and thus assuming enough regularity, we can rewrite $\hat{s}_i(t)$ as the integral of exponential signals $e^{\tau t}$ through inverse Laplace transform. It can be shown that when the input to system \eqref{eqn:state_space} takes the form $e^{\tau t}M^{1/2}v_i$, the steady-state frequency trajectory $\ol{\omega}(t)$ is given by $H_i(\tau) e^{\tau t}M^{-1/2}v_i$, where $H_i(\tau)$ is a complex-valued function of $\tau$ specifying the system gain and phase shift. We refer to the function $H_i(\tau)$ as the $i$-th spectral transfer function. Compared to classical transfer functions, the spectral version does not capture the relationship between any input-output pair, but in contrast captures the behavior of system \eqref{eqn:state_space} from a network perspective. Once the spectral transfer functions are known, we can compute the steady-state trajectories for general input signal $s(t)$ through the following synthesis formula
$$
\ol{\omega}(t)=\sum_{i=1}^n \calL^{-1}\set{H_i(\tau)\calL\set{\hat{s}_i(t)}(\tau)}M^{-1/2}v_i
$$

\begin{thm}\label{thm:laplace}
For each $i$, assuming $\gamma^2-4\lambda_i\neq 0$, the $i$-th spectral transfer function is given by
$$
H_i(\tau)=\frac{\tau}{\tau^2+\gamma\tau + \lambda_i}
$$
$ $
\end{thm}

\begin{proof}
\iftoggle{isreport}{
Please refer to Appendix \ref{proof:laplace}.
}{
Please refer to our online report \cite{report}.
}
\end{proof}

We remark that a similar formula also shows up in \cite{paganini2017global} as the representative machine transfer function for swing dynamics.

\section{Interpretations}\label{section:interpretations}
In this section, we present a collection of intuition that can be devised from the results in Section \ref{section:characterization}. They are useful for making general inferences and for the controller design in Section \ref{section:design}.
\subsection{Network connectivity and system stabilization}\label{subsection:mixing}
We first clarify how the network connectivity affects the system stability. Towards this goal, we rewrite \eqref{eqn:omega_evolution} as
$$
\omega(t)=\sum_{i=1}^n\hat{s}_i \hat{\omega}^i(t)M^{-1/2}v_i
$$
The signal $\hat{\omega}^i(t)$ captures the response of system \eqref{eqn:state_space} along $M^{-1/2}v_i$ to a step function input. By Theorem \ref{thm:modes}, we see that whether the system oscillates or not is determined by the signs of $\gamma^2-4\lambda_i$. For $\lambda_i$ such that $\gamma^2-4\lambda_i> 0$, we have
$$
\hat{\omega}^i(t)=\frac{1}{\sqrt{\gamma^2-4\lambda_i}}\paren{e^{\phi_{i,+}t}-e^{\phi_{i,-}t}}
$$
with $\phi_{i,\pm}\le 0$. Thus the system is over-damped along $M^{-1/2}v_i$, and deviations along $M^{-1/2}v_i$ exponentially fades away without oscillation. The slower-decaying exponential has a decaying rate determined by $\phi_{i,+}$, which is a decreasing function in $\lambda_i$. Thus a larger $\lambda_i$ implies faster decay. Intuitively, this tells us that when the system damping is strong with respect to its inertia, adding connectivity helps move more disturbances to the damping component so that disturbances can be absorbed sooner.

For $\gamma$ such that $\gamma^2-4\lambda_i<0$, we have
$$
\hat{\omega}^i(t)=\frac{2}{\sqrt{4\lambda_i-\gamma^2}}e^{-\frac{\gamma}{2}t}\sin\paren{\frac{\sqrt{4\lambda_i-\gamma^2}}{2}t}
$$
Thus the system is under-damped along $M^{-1/2}v_i$ and oscillations do occur. We also note that larger values of $\lambda_i$ lead to oscillations of higher frequency.
This intuitively can be interpreted as the following: When the system damping is not strong enough compared to its inertia, adding connectivity causes the un-absorbed oscillations to propagate throughout the network faster, bringing disturbances to the already over-burdened damping components, making the system oscillate in a higher frequency.

We thus see that Theorem \ref{thm:modes} and Theorem \ref{thm:response} precisely clarify our seemingly contradictory intuition on whether connectivity is beneficial to stabilization - it depends on how strong the system is damped compared to its inertia, i.e.~how fast the system can dissipate energy.
\subsection{Robustness to disturbance}\label{subsection:laplace}
The impact of different system parameters in the Laplace domain can be understood from the spectral transfer functions $H_i$. Recall by Theorem \ref{thm:laplace}, for a signal of the form $s(t)=e^{\tau t}v_i$, the steady state output signal of \eqref{eqn:state_space} is
$$
\ol{\omega}_i(t;\tau)=\frac{\tau e^{\tau t}}{\tau^2+\gamma \tau + \lambda_i}M^{-1}v_i
$$
In particular, if we focus on the $j$-th component of $\ol{\omega}_i(t)$, which corresponds to the frequency trajectory of bus $j$, we have
$$
\ol{\omega}_{i,j}(t;\tau)=\frac{\tau v_{i,j} e^{\tau t}}{M_j\tau^2+D_j\tau +\lambda_i M_j}
$$
Under the proportional rating assumption mentioned in Section \ref{section:model}, one can show that $\lambda_i M_j =\ol{\lambda}_i$, where $\ol{\lambda}_i$ is the $i$-th Laplacian eigenvalue when the ``heaviest'' generator is normalized to have unit inertia $\max_{j\in\calN}M_j=1$ and can be interpreted as the pure topological part in the Laplacian eigenvalues $\lambda_i$. This allows us to compute
\begin{equation}\label{eqn:laplace_gain}
\abs{\ol{\omega}_{i,j}(\mathrm{j}\sigma)}=\frac{\abs{\sigma}}{\sqrt{M_j^2\sigma^4+(D_j^2-2\ol{\lambda_i}M_j)\sigma^2 + \ol{\lambda}_i^2}}
\end{equation}
and conclude the following (See Fig.~\ref{fig:spectral} for an illustration): a) For high frequency signals, the gain can be approximated by $\abs{\ol{\omega}_{i,j}(\mathrm{j}\sigma)}\approx \frac{1}{M_j\sigma}$ and therefore the key parameter to suppress such disturbence is the rotational inertia $M_j$; b) For small frequency signals, the gain approximates to $\abs{\ol{\omega}_{i,j}(\mathrm{j}\sigma)}\approx \frac{\sigma}{\ol{\lambda}_i}$ and hence low frequency disturbences are mostly suppressed by the network topology; c) For any fixed frequency $\sigma$, $\abs{\ol{\omega}_{i,j}(\mathrm{j}\sigma)}$ is decreasing in $D_j$. This means that a larger damping leads to smaller gains for all frequencies. Such decrease, however, is negligible for very large or very small $\sigma$, and therefore increasing $D_j$ mostly helps the system to suppress oscillations in the medium frequency band.
\begin{figure}
\centering
\includegraphics[scale=.4]{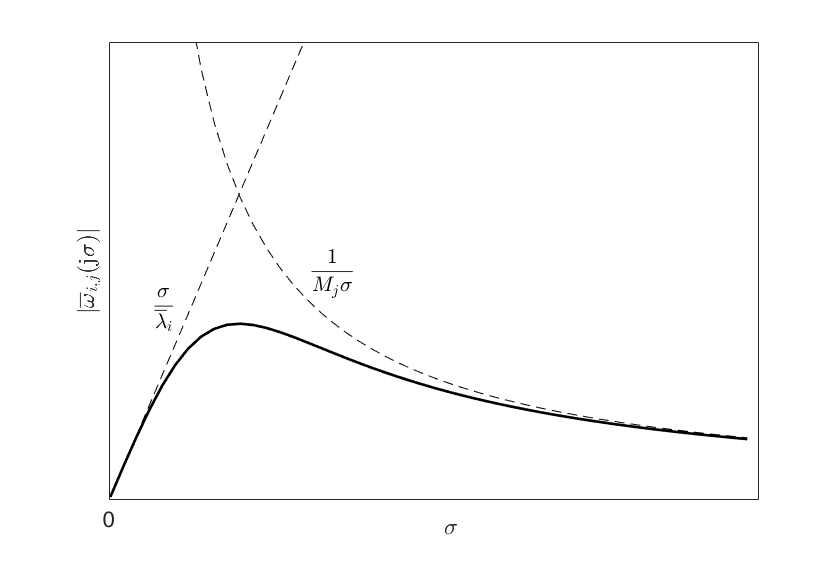}
\caption{Illustration of the gain $|\ol{\omega}_{i,j}(\mathrm{j}\sigma)|$ as a funciton of $\sigma$.}\label{fig:spectral}
\end{figure}
\subsection{Impact of damping}\label{subsection:general}
As a by-product of our study, we can also examine how the system damping impacts the system performance. Towards this goal, we study two common metrics for $\hat{\omega}^i(t)$: a) settling time, which is the time it takes $\hat{\omega}^i(t)$ to get within a certain range\footnote{The range is specified as $[\omega_i^*-c, \omega_i^*+c]$, where $\omega^*$ is the equilibrium state and $c$ is a constant.} around the steady state; b) nadir, which is defined to be the sup norm of $\hat{\omega}^i(t)$. 
 Table \ref{table:system_metrics} summarises the formulae\footnote{We define $\Delta_i=\abs{\gamma^2-4\lambda_i}$ to simplify the formulae. The settling time formula is an upper bound as finding its exact value requires solving transcendental equations, which is generally hard.} for these metrics, 
and one can show using basic calculus that both the settling time and nadir are decreasing functions of $\gamma$, and thus decreasing in the damping constants $D_j$ (provided that the inertia constants $M_j$ are fixed).
\begin{table*}
\centering
\caption{System performance in terms of network Laplacian eigenvalues, generator inertia and damping.}\label{table:system_metrics}
\begin{tabular}{|c |c | c |}
\hline
Case & Settling Time & Nadir \\
\hline
$\gamma^2>4\lambda_i$ & $\frac{1}{\gamma-\sqrt{\Delta_i}}\ln\paren{\frac{1}{4c^2\Delta_i}}$
&
$
\frac{1}{\sqrt{\Delta_i}}\Bigg[\paren{\frac{\gamma+\sqrt{\Delta_i}}{\gamma-\sqrt{\Delta_i}}}^{\frac{-\gamma+\sqrt{\Delta_i}}{2\sqrt{\Delta_i}}}
-\paren{\frac{\gamma+\sqrt{\Delta_i}}{\gamma-\sqrt{\Delta_i}}}^{\frac{-\gamma-\sqrt{\Delta_i}}{2\sqrt{\Delta_i}}}\Bigg]
$\\
\hline
$\gamma^2<4\lambda_i$ &$\frac{1}{\gamma}\ln\paren{\frac{4}{c^2\Delta_i}}$ & $\frac{2}{\sqrt{\Delta_i}}\exp\paren{-\frac{2\pi \gamma}{\sqrt{\Delta_i}}}$\\
\hline
\end{tabular}
\end{table*}

This result of course does not generalize to $\omega(t)$ in a straightforward way because of the possibility of negative $\hat{s}_i$. Instead of focusing on $\omega(t)$ for a specific $s(t)$, we can look at all possible $\omega(t)$ and generalize our previous interpretations to the worst case performance metric. To be concrete, let us take nadir as an example. By \eqref{eqn:omega_evolution}, we see the nadir of frequency trajectory at bus $j$ satisfies
\begin{IEEEeqnarray*}{rCl}
\norm{\omega_j(t)}_{\infty} &\le& M_j^{-1/2}\sum_{i=1}^n\abs{\hat{s}_i}\norm{v_{i,j}\hat{\omega}^i(t)}_{\infty}\\&\le& M_j^{-1/2}\sqrt{\sum_{i=1}^n\abs{\hat{s}_i}^2}\sqrt{\sum_{i=1}^n\norm{v_{i,j}\hat{\omega}^i(t)}_{\infty}^2}\\
&=& M_j^{-1/2}\norm{M^{-1/2}s}_2\sqrt{\sum_{i=1}^n\norm{v_{i,j}\hat{\omega}^i(t)}_{\infty}^2}\\
&:=&M_j^{-1/2}\norm{M^{-1/2}s}_2\norm{\omega}_\infty^w
\end{IEEEeqnarray*}
It is easy to see that all the inequalities above can attain equalities. Therefore among all input $s$ with scaled unit energy $\norm{M^{-1/2}s}_2=1$, the worst possible nadir is $M_j^{-1/2}\norm{\omega}_\infty^w$, which is a decreasing function of $\delta$ from our previous discussions.

This worst case nadir is a system level metric that is independent of the input. Although this metric does not predict the exact nadir for any specific input, it does reveal to what extent the system can tolerate disturbances of certain energy, which is a property that is intrinsic to the system itself. Moreover, for secure and robust operation of the grid, we need to make sure that the worst case nadir is well-controlled. Similar argument can be also applied to the settling time for $\omega_j(t)$.
\subsection{System tradeoffs}
When choosing system control parameters, there are usually tradeoffs among different performance goals and we must balance different aspects to obtain a good design. A key tradeoff of this type revealed by our previous disucssions is the tradeoff between having small network intrinsic frequency and improving system robustness against low frequency disturbance.

More specifically, it is easy to show that $\abs{\ol{\omega}_{i,j}(\mathrm{j}\sigma)}$ is maximized at $\sigma^*_{i,j}=\sqrt{\frac{\ol{\lambda}_i}{M_j}}$. In other words $\sigma^*_{i,j}$ can be interpreted as an intrinsic frequency of the network and oscillations around $\sigma^*_{i,j}$ are amplified at bus $j$ through the transmission system. Typically high frequency oscillations should be suppressed and thus we want smaller $\sigma^*_{i,j}$, which in turn leads to smaller $\ol{\lambda}_i$. On the other hand, we have shown that in order for the system \eqref{eqn:state_space} to be robust against low frequency noise (such as periodic load oscillations within a day), the transmission network should be designed with as large connectivity $\ol{\lambda}_i$ as possible. As a result, we cannot make the system \eqref{eqn:state_space} having small intrinsic frequency and being robust against low frequency noise at the same time.

\section{Controller Design for Load-Side Participation}\label{section:design}
In this section, we discuss two implications of our results in Section \ref{section:interpretations} to load-side controller design.
\subsection{Benefits of load-side participation}
We adopt the controller design from \cite{zhao2014design} as an example to explain the benefits of load-side participation. We assume the system deviation is small so that the capacity bounds of load side controllers are not binding. In this setting, the control law of \cite{zhao2014design} is simplified to
\begin{equation}\label{eqn:control_law}
d_j=K_p\omega_j
\end{equation}
which when plugged into \eqref{eqn:state_space} can be absorbed into the damping term $D_j \omega_j$. Therefore the integration of controller \eqref{eqn:control_law} effectively increases the system damping level. Based on our discussions in Section \ref{section:interpretations}, we conclude that load-side participation decreases both the settling time and nadir of \eqref{eqn:state_space}. This means with load-side participation, the system \eqref{eqn:state_space} is more responsive and its nadir under a disturbance is also better controlled.

Such benefits have been observed and confirmed in a series of work \cite{zhao2012frequency,zhao2012swing,zhao2014design,mallada2014optimal,zhao2016unified} in their simulations. With our framework, it is possible to analytically derive such results and quantify how beneficial the load-side integration can be when we use a certain system gain $K_p$. Moreover, it is observed in \cite{zhao2012frequency} that load-side participation also helps maintain system stability when the generator output fluctuates. Using our characterization in the Laplace domain, we see that such benefit comes from the improved system ability in suppressing oscillations of medium band frequency.
\subsection{Proportional-Derivative (PD) controller}
Despite the many benefits of load-side controllers we have explained so far, one component still missing in \eqref{eqn:control_law} is that they only affect the system damping but cannot increase the system inertia. As we mentioned in Section \ref{section:interpretations}, the system inertia is the key parameter affecting the system robustness against high frequency oscillations. 
Nevertheless, a quick glimpse to \eqref{eqn:state_space} suggests that in order to have larger $M_j$, it suffices to add a derivative term in \eqref{eqn:control_law}, which can be implemented through power electronics or invertors \cite{mallada2016idroop}:
\begin{equation}\label{eqn:new_control_law}
d_j=K_p\omega_j+K_d\dot{\omega}_j
\end{equation}
Although it is a natural idea to generalize proportional controllers to PD controllers for performance tuning, we see that the need of this derivative term can actually be reversed engineered from our characterizations. Moreover, our framework reveals how the parameters $K_p$ and $K_d$ affect the system performance precisely, allowing us to optimize such gains subject to different design goals.

Using derivative terms in controller design is often problematic in practice due to the amplified noise in its measurement. However, we know from Section \ref{section:interpretations} that neither adding damping nor increasing network connectivity is particularly effective in suppressing disturbences in the high frequency regime. Thus in order to improve the grid stability under high frequency fluctuations, having certain components of the network that are able to measure the signal derivatives either explicitly or implicitly to provide the necessary inertia is inevitable. We are still investigating the optimal tradeoff between the increased virtual inertia yet also the amplified noise from using PD controllers.  

\section{Evaluation}\label{section:evaluation}
In this section, we simulate the controller design \eqref{eqn:new_control_law} over the IEEE 39-bus New England interconnection system as shown in Fig.~\ref{fig:line_diagram} and compare its performance to that of \eqref{eqn:control_law} and the conventional droop control. There are 10 generators and 29 load nodes in the system and we take the system parameters from the Matpower Simulation Package \cite{zimmerman2011matpower}. In contrast to our theoretical analysis, the simulation data does not satisfy the proportional rating assumption in Section \ref{section:model}. The droop control is implemented as the $D_j\omega_j$ term for generator buses and is deactivated for simulations with the controllers \eqref{eqn:control_law} and \eqref{eqn:new_control_law}. We assume all the buses (including the generator buses) have load-side participation enabled and pick the controller gains $K_p$ and $K_d$ heterogeneously in proportional to the bus damping $D_j$.
\begin{figure}
\centering
\includegraphics[scale=.85]{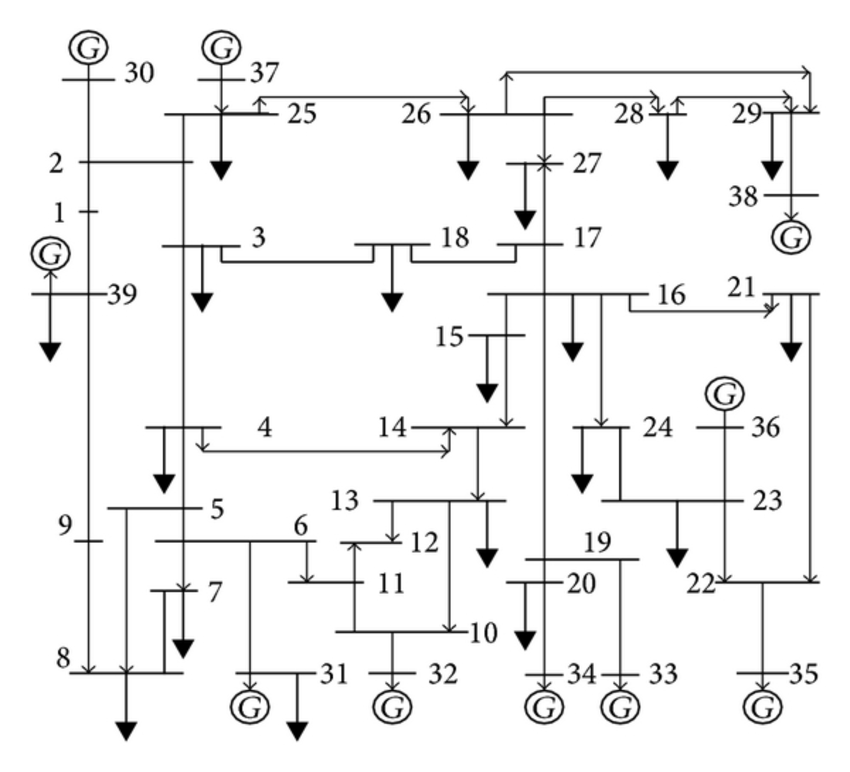}
\caption{Line diagram of the IEEE 39-bus interconnection testbed.}\label{fig:line_diagram}
\end{figure}
\subsection{Robustness against measurement noise}
We first look at the controller performances against measurement noise. Towards this goal, we add a white Gaussian measurement noise of power $-20$ dBW to the frequency sensor at bus 30 and observe its frequency trajectory, which is shown in Fig.~\ref{fig:white_noise}. We can see that the controller \eqref{eqn:control_law} is less prone to measurement noise compared to the conventional droop control, because it increases the system damping level and therefore helps suppress the medium frequency part of the noise. However, its benefit in suppressing high frequency noise is limited, as one can see from its performance gap compared with the controller \eqref{eqn:new_control_law}. To more clearly see such distinction, we replace the measurement noise at bus 30 with the signal $0.2\sin(10\pi t)$ p.u.~that contains only high frequency component and observe its trajectory. The result is shown in Fig.~\ref{fig:sine}. In this case, we see that controller \eqref{eqn:control_law} performs nearly the same as the conventional droop control, while the system under the improved controller \eqref{eqn:new_control_law} exhibits much smaller oscillation.
\begin{figure}
\centering
\includegraphics[scale=.6]{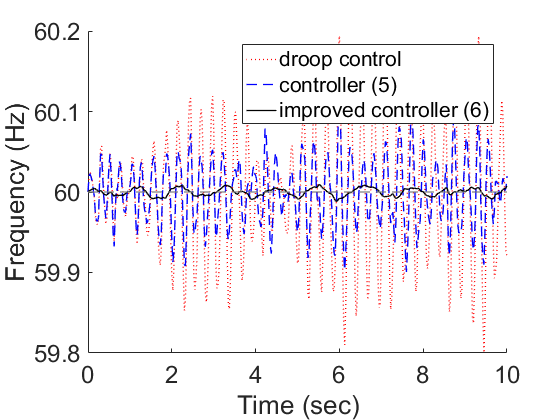}
\caption{Frequency trajectory at bus 30 when we add white Gaussian measurement noise of $-20$ dBW.}\label{fig:white_noise}
\end{figure}
\begin{figure}
\centering
\includegraphics[scale=.6]{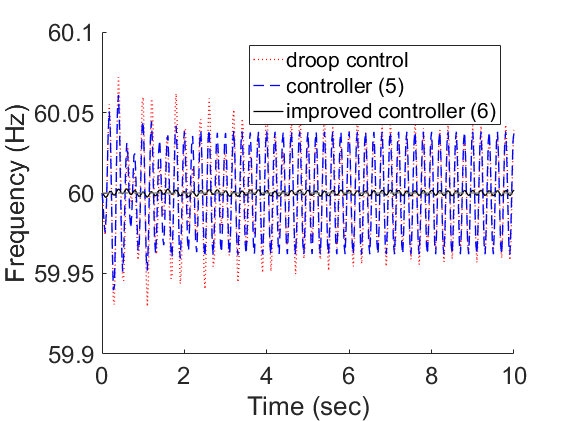}
\caption{Frequency trajectory at bus 30 when we add a signal following the sine curve $0.2\sin(10\pi t)$ p.u.}\label{fig:sine}
\end{figure}
\subsection{Wind power data}
Next, we look at the performance of the controllers under real wind power generation data from \cite{xu2011demand}. We choose bus 30 to be the wind generator, whose output follows the profile given in \cite{xu2011demand} and look at the frequency trajectory at bus 36. The two buses are specifically chosen to be geographically far away so that the simulation results reflect end user perception of such renewable penetration. The simulation result is shown in Fig.~\ref{fig:wind}. As one can see, compared to controller \eqref{eqn:control_law}, the improved controller \eqref{eqn:new_control_law} incurs smaller frequency deviation at almost all time, and the resulting trajectory is smoother. This is because \eqref{eqn:new_control_law} filters away high frequency fluctuations in the generator profile. We expect such benefit to be more significant when the  system aggregate load fluctuates more frequently because of increasing renewable penetration.
\begin{figure}
\centering
\includegraphics[scale=.6]{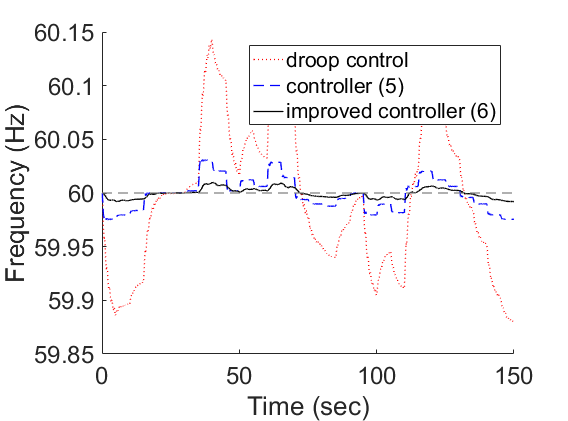}
\caption{Frequency trajectory at bus 36 under wind power output at bus 30.}\label{fig:wind}
\end{figure}

\section{Conclusion}\label{section:conclusion}
In this work, we proposed a framework using spectral graph theory that captures the interplay among different system parameters. It leads to precise characterizations on how control parameters affect the system performance and allows us to make general inferences without extensive simulation. We quantified the benefits of load-side participation within this framework and explained how we can improve the controller design so that the system is more robust against high frequency oscillations.

We remark that our framework can be generalized to include secondary frequency controllers. In particular, load-side controllers for secondary frequency regulation can usually be locally interpreted as the control law \eqref{eqn:control_law} plus a term that captures the overall supply-demand imbalance from other parts of the network. Therefore in terms of system stabilization, all of our discussion about how load-side participation helps system \eqref{eqn:state_space} in both transient and steady state will still apply. In terms of driving the system back to the nominal state, the framework explains how the cyber and physical network topologies interact with each other and suggest methods to improve the overall system convergence rate. Due to space limitation, we refer interested readers to \cite{guo2018cyber} for more detailed discussions. We are still investigating how our results can be generalized to more detailed models (say where the generators have higher order or nonlinear dynamics).

\bibliographystyle{IEEEtran}
\bibliography{biblio}

\iftoggle{isreport}{
\appendix
\subsection{Proof of Theorem \ref{thm:modes}}\label{proof:modes}
By Schur complement, we can compute the characteristic polynomial of $A$ as
\begin{IEEEeqnarray*}{rCl}
&&\det(A-tI)\\
&=&\det(-tI_m)\det\paren{-\gamma I_n-tI_n - \frac1t M^{-1}CBC^T}\\
&=&\det(-tI_m)\det\paren{M^{-1/2}}\det\paren{M^{1/2}}\\
&&\times\det\paren{M^{1/2}(-\gamma I_n-tI_n - \frac1t M^{-1}CBC^T)M^{-1/2}}\\
&=&(-1)^{m+n}t^{m-n}\det\paren{L+(\gamma t+ t^2)I_n}\nonumber
\end{IEEEeqnarray*}
All the above algebra is understood to be over the polynomial field generated from $\R[t]$ and thus we do not need to assume $t\neq 0$.

The term $t^{m-n}$ contributes $m-n$ multiplicity to the eigenvalue 0 (in the case $\calG$ is a tree, or equivalently $m=n-1$, this is understood to mean that $t^{m-n}=t^{-1}$ cancels one multiplicity of 0). Let us now tackle the factor $\det\paren{L+(\gamma t+t^2)I_n}$. It is easy to see that $\det\paren{L+(\gamma t+ t^2)I_n}=0$ if and only if
$$
t^2 + \gamma t +\lambda= 0
$$
for some eigenvalue $\lambda$ of $L$. Therefore the roots of $\det\paren{L+(\gamma t+t^2)I_n}$ are given as
\begin{equation}\label{eqn:roots_of_system_matrix}
\phi = \frac{-\gamma \pm \sqrt{\gamma^2-4\lambda}}{2}
\end{equation}
with $\lambda$ traversing all eigenvalues of $L$. Among these roots, $0$ appears exactly once, coming from the zero eigenvalue of $L$. Thus altogether, we know the eigenvalue of $A$ consists of $0$ with multiplicity $m-n+1$ and non-zero roots of the form given by \eqref{eqn:roots_of_system_matrix}.

Next we determine the eigenvectors of the system matrix $A$. Let $\phi\neq 0$ be an eigenvalue corresponding to an eigenvector $[\omega; P]$. Then we have
\begin{subequations}
\begin{IEEEeqnarray}{rCl}
-\gamma\omega - M^{-1}CP&=&\phi\omega\label{eqn:first_eig_const}\\
BC^T\omega &=& \phi P \label{eqn:second_eig_const}
\end{IEEEeqnarray}
\end{subequations}
Substituting \eqref{eqn:second_eig_const} to \eqref{eqn:first_eig_const} and multipliying $M^{1/2}$ on both sides, we see that
$$
LM^{1/2}\omega = -(\phi^2+\gamma\phi)M^{1/2}\omega
$$
or in other words, $M^{1/2}\omega$ is an eigenvector of $L$ affording $-(\phi^2+\gamma\phi)=\lambda$. For any such $\omega$, the corresponding $P$ by \eqref{eqn:second_eig_const} is given by $P=\phi^{-1}BC^T\omega$. Moreover, we see that $\phi=-\gamma$ is a simple eigenvalue of $A$ as the corresponding $\lambda=0$ is a simple eigenvalue of $L$. Note that
$$
\norm{B^{1/2}C^TM^{-1/2}v_1}^2 = v_1^TLv_1 = 0
$$
implies $BC^TM^{-1/2}v_1=0$, and therefore we see that
$$
\bracket{M^{-1/2}v_1; -\frac{1}{\gamma} BC^TM^{-1/2}v_1}=\bracket{M^{-1/2}v_1;0}
$$
is an eigenvector of $A$ affording $\phi=-\gamma$.

For $\phi=0$, from \eqref{eqn:second_eig_const} we have $\omega = cI_n$ for some $c$. Plugging back to \eqref{eqn:first_eig_const}, we have
$$
cI_n = -\frac{1}{\gamma}M^{-1}CP 
$$
and therefore $c\bff{1}^TMI_n=0$ which implies $c=0$. This then implies $P\in\kernel(C)$ and $\omega=0$. Therefore the eigenvectors corresponding to $\phi=0$ are given by $[0; P]$ with $P\in \kernel(C)$, which by dimension theorem has dimension $m - \rank(C)=m-n+1$.
\qedd

\subsection{Proof of Theorem \ref{thm:response}}\label{proof:response}
Recall we have shown in Section \ref{section:characterization} that $A$ is diagonalizable over the complex field $\C$, provided critical damping does not occur. Let $A=Q\Lambda Q^{-1}$ be an eigenvalue decomposition of $A$. Then we have $e^{At}=Qe^{\Lambda t}Q^{-1}$ for any $t\in\R$. Now the solution to the system \eqref{eqn:state_space} with a constant input $s$ and nominal initial state is given as
\begin{IEEEeqnarray}{rCl}
x(t)&=&\int_0^t \paren{e^{A(t-\tau)}\begin{bmatrix}
M^{-1}s\\0
\end{bmatrix}}d\tau \nonumber \\
&=&Q\int_0^t e^{\Lambda(t-\tau)}d\tau Q^{-1}\begin{bmatrix}
M^{-1}s\\0
\end{bmatrix}\label{eqn:general_solution}
\end{IEEEeqnarray}
Write $\Lambda$ in block diagonal form as
$$
\Lambda = \begin{bmatrix}
0 & 0\\
0 & \Phi
\end{bmatrix}
$$
where $\Phi$ collect all nonzero eigenvalues of $A$. By Theorem \ref{thm:modes}, we can compute \eqref{eqn:general_solution} to
\begin{IEEEeqnarray}{rCl}
&&x(t)\nonumber\\
&=&Q\int_0^t\begin{bmatrix}
(t-\tau)I_{m-n+1} & 0 \\
0 & e^{\Phi (t-\tau)}
\end{bmatrix}d\tau Q^{-1}\begin{bmatrix}
M^{-1}s\\0
\end{bmatrix}\nonumber\\
&=&Q\begin{bmatrix}
\frac{t^2}{2}I_{m-n+1} & 0 \\
0 & \Phi^{-1}(e^{\Phi t}-I_{2n-1})
\end{bmatrix} Q^{-1}\nonumber\\
&&\times\begin{bmatrix}
M^{-1}s\\0
\end{bmatrix}
\label{eqn:general_solution_sim}
\end{IEEEeqnarray}
Consider an eigen-pair $(\lambda_i, v_i)$ of $L$. For $i=1$, we have $\lambda_1=0$ and therefore by Theorem \ref{thm:modes}, we see
$
\bracket{M^{-1/2}v_1; 0}
$
is an eigenvector of $A$ affording $-\gamma$. For $i\ge 2$, we know
$$
\phi_{i,+}:=\frac{-\gamma+\sqrt{\gamma^2-4\lambda_i}}{2} \quad \phi_{i,-}:=\frac{-\gamma-\sqrt{\gamma^2-4\lambda_i}}{2}
$$
are eigenvalues of $A$ with corresponding eigenvectors $[M^{-1/2}v_i; \phi_{i,\pm}^{-1}BC^T M^{-1/2}v_i]=:z_{i,\pm}$. This allows us to decompose
\begin{IEEEeqnarray*}{rCl}
\begin{bmatrix}
M^{-1/2}v_i\\
0
\end{bmatrix} &=& \frac{\sqrt{\gamma^2-4\lambda_i}-\gamma}{2\sqrt{\gamma^2-4\lambda_i}}\begin{bmatrix}
M^{-1/2}v_i\\
\phi_{i,+}^{-1}BC^TM^{-1/2}v_i
\end{bmatrix}\\
&& + \frac{\sqrt{\gamma^2-4\lambda_i}+\gamma}{2\sqrt{\gamma^2-4\lambda_i}}\begin{bmatrix}
M^{-1/2}v_i\\
\phi_{i,-}^{-1}BC^TM^{-1/2}v_i
\end{bmatrix}\\
&=:&\lambda_{i,+}z_{i,+}+ \lambda_{i,-}z_{i,-}
\end{IEEEeqnarray*}
which then implies
\begin{IEEEeqnarray*}{rCl}
\begin{bmatrix}
M^{-1}s\\0
\end{bmatrix} &=& \sum_{i=1}^n\hat{s}_i\begin{bmatrix}
M^{-1/2}v_i\\0
\end{bmatrix}\\
&=&\hat{s}_1
\begin{bmatrix}
M^{-1/2}v_1\\ 0
\end{bmatrix} + \sum_{i=2}^n\lambda_{i,+}\hat{s}_iz_{i,+}\\
&&+ \sum_{i=2}^n\lambda_{i,-}\hat{s}_iz_{i,-}\end{IEEEeqnarray*}
We emphasize that the input to system \eqref{eqn:state_space} is $s$, and $[M^{-1}s;0]$ is the signal obtained by multiplying the input scaling matrix $[M^{-1};0]$ to $s$. By linearity, we can compute \eqref{eqn:general_solution_sim} as
\begin{IEEEeqnarray*}{rCl}
x(t)&=&-\frac{\hat{s}_1}{\gamma}e^{-\gamma t}\begin{bmatrix}
M^{-1/2}v_1\\ 0
\end{bmatrix} + \sum_{i=2}^n\frac{\lambda_{i,+}\hat{s}_i}{\phi_{i,+}}e^{\phi_{i,+}t}z_{i,+}\\
&&+ \sum_{i=2}^n\frac{\lambda_{i,-}\hat{s}_i}{\phi_{i,-}}e^{\phi_{i,-}t}z_{i,-}+\frac{\hat{s}_1}{\gamma}\begin{bmatrix}
M^{-1/2}v_1\\ 0
\end{bmatrix} \\
&&- \sum_{i=2}^n\frac{\lambda_{i,+}\hat{s}_i}{\phi_{i,+}}z_{i,+}- \sum_{i=2}^n\frac{\lambda_{i,-}\hat{s}_i}{\phi_{i,-}}z_{i,-}
\end{IEEEeqnarray*}
One can check by direct computation that for $i\ge 2$,
$$
\frac{\lambda_{i,+}}{\phi_{i,+}}+\frac{\lambda_{i,-}}{\phi_{i,-}}=0
$$
Therefore when restricting to the $\omega$ part in $x$, we have
\begin{IEEEeqnarray*}{rCl}
M^{1/2}\omega(t)
&=&\frac{\hat{s}_1}{\gamma}v_1-\frac{\hat{s}_1}{\gamma}e^{-\gamma t}v_1
+ \sum_{i=2}^n\frac{\lambda_{i,+}\hat{s}_i}{\phi_{i,+}}e^{\phi_{i,+}t}v_i\\
&&+ \sum_{i=2}^n\frac{\lambda_{i,-}\hat{s}_i}{\phi_{i,-}}e^{\phi_{i,-}t}v_i
\end{IEEEeqnarray*}
This together with the fact
$$
\frac{\lambda_{i,\pm}}{\phi_{i,\pm}}=\pm\frac{1}{\sqrt{\gamma^2-4\lambda_i}}
$$
and the observation that
$$
\frac{\hat{s}_1}{\gamma}v_1-\frac{\hat{s}_1}{\gamma}e^{-\gamma t}v_1=\frac{\hat{s}_1}{\sqrt{\gamma^2-4\lambda_1}}\paren{e^{\phi_{1,+}t}-e^{\phi_{1,-}t}}v_1
$$
completes the proof.
\qedd

\subsection{Proof of Theorem \ref{thm:laplace}}\label{proof:laplace}
First consider $i\ge 2$. From the proof of Theorem \ref{thm:response}, we know 
\begin{equation}\label{eqn:eigen_decomposition}
\begin{bmatrix}
M^{-1/2}v_i\\0
\end{bmatrix}=\lambda_{i,+}z_{i,+}+\lambda_{i,-}z_{i,-}
\end{equation}
This together with the calculation in \eqref{eqn:general_solution} implies that for input signal of the form $s_i(t)M^{1/2}v_i$, the system response of \eqref{eqn:state_space} is given as
\begin{IEEEeqnarray*}{rCl}
x(t)&=&\lambda_{i,+}\int_0^t e^{\phi_{i,+}(t-\tau')}s_i(\tau')d\tau' z_{i,+}\\
&& +\lambda_{i,-}\int_0^t e^{\phi_{i,-}(t-\tau')}s_i(\tau')d\tau' z_{i,-}
\end{IEEEeqnarray*}
For $s_i(t)=e^{\tau t}$, we then have
\begin{IEEEeqnarray*}{rCl}
x(t)&=&\frac{\lambda_{i,+}}{\tau-\phi_{i,+}}\paren{e^{\tau t} - e^{\phi_{i,+}t}}z_{i,+}\\
&& +\frac{\lambda_{i,-}}{\tau-\phi_{i,-}}\paren{e^{\tau t} - e^{\phi_{i,-}t}}
z_{i,-}
\end{IEEEeqnarray*}
and therefore when restricting to the frequency trajectory, we have
\begin{IEEEeqnarray*}{rCl}
M^{1/2}\omega(t)&=&\frac{\lambda_{i,+}(\tau-\phi_{i,-})+\lambda_{i,-}(\tau-\phi_{i,+})}{\tau^2-(\phi_{i,+}+\phi_{i,-})\tau +\phi_{i,+}\phi_{i,-}}e^{\tau t}v_i\\
&&-\paren{\frac{\lambda_{i,+}e^{\phi_{i,+}t}}{\tau-\phi_{i,+}} + \frac{\lambda_{i,-}e^{\phi_{i,-}t}}{\tau-\phi_{i,-}}}v_i
\end{IEEEeqnarray*}
Noting $\lambda_{i,+}+\lambda_{i,-}=1$, $\lambda_{i,+}\phi_{i,-}+\lambda_{i,-}\phi_{i,+}=0$, $\phi_{i,+}+\phi_{i,-}=-\gamma$ and $\phi_{i,+}\phi_{i,-}=\lambda_i$ and dropping transient terms, we see
$$
\ol{\omega}(t)=\frac{\tau}{\tau^2+\gamma\tau + \lambda_i}e^{\tau t} M^{-1/2}v_i
$$
For $i=1$, we do not need to decompose the signal as in \eqref{eqn:eigen_decomposition} and a similar calculation leads to
$$
\ol{\omega}(t)=\frac{e^{\tau t}M^{-1/2}v_1}{\tau+\gamma}=\frac{\tau}{\tau^2+\gamma\tau + \lambda_1}e^{\tau t} M^{-1/2}v_1
$$
where the last equality is because $\lambda_1=0$.
\qedd
}{}

\end{document}